\documentclass[conference]{IEEEtran}
\IEEEoverridecommandlockouts
\usepackage{graphicx}
\usepackage{amssymb}
\usepackage{amsmath}
\usepackage{mathtools}
\usepackage{cite}
\usepackage{subfigure}
\usepackage{mathrsfs}
\usepackage[displaymath,mathlines]{lineno}
\usepackage{color}
\usepackage{multirow}
\usepackage{algpseudocode}
\usepackage{algorithm,algpseudocode}
\usepackage{algorithmicx}
\usepackage{pbox}
\usepackage{multicol}
\usepackage{lipsum}
\usepackage{amsthm}
\usepackage{relsize}
\usepackage{lipsum}
\usepackage{enumitem}
\usepackage{epstopdf}
\usepackage{microtype}
\usepackage{iftex}
\ifLuaTeX\usepackage[utf8]{luainputenc}\else\usepackage[utf8]{inputenc}\fi
\usepackage[OT1]{fontenc}

\newcommand{\doublewidetilde}[1]{{%
		\mathpalette\double@widetilde{#1}%
	}}
	
	\makeatletter
	
	\def\BState{\State\hskip-\ALG@thistlm}
	\makeatother

	\newtheorem{theorem}{Theorem}
	
	\newtheorem{corollary}{Corollary}

	\newcounter{eqnback}
	\newcounter{eqncnt}

	\newcommand{\tr}{\ensuremath{\mathsf{T}}}
	\newcommand{\conjtr}{\ensuremath{\mathsf{H}}}
	\newcommand{\Exp}{\ensuremath{\mathbb{E}}}
	\newcommand{\abb}[1]{#1} 
	\newcommand{\defas}{\triangleq}

\begin{document}
		%
		\title{Two-Layer Decoding in Cellular Massive MIMO Systems with Spatial Channel Correlation}


	\author{\IEEEauthorblockN{Trinh~Van~Chien, Christopher~Moll\'{e}n, Emil~Bj\"{o}rnson}
		\IEEEauthorblockA{Department of Electrical
			Engineering (ISY), Link\"{o}ping University, SE-581 83 Link\"{o}ping, Sweden\\
			\{trinh.van.chien, emil.bjornson\}@liu.se, chris.mollen@gmail.com}
		\thanks{This paper was supported by the European Union's Horizon 2020 research and innovation programme under grant agreement No 641985 (5Gwireless). It was also supported by ELLIIT and CENIIT.}
	}
		
		\maketitle
		
		\begin{abstract}
		 This paper studies a two-layer decoding method that mitigates inter-cell interference in multi-cell Massive \abb{MIMO} systems. In layer one, each base station (\abb{BS}) estimates the channels to intra-cell users and uses the estimates for local decoding on each \abb{BS}, followed by a second decoding layer where the \abb{BS}s cooperate to mitigate inter-cell interference. An uplink achievable spectral efficiency (\abb{SE}) expression is computed for arbitrary two-layer decoding schemes, while a closed-form expression is obtained for correlated Rayleigh fading channels, maximum-ratio combining (MRC), and large-scale fading decoding (\abb{LSFD}) in the second layer.  We formulate a non-convex sum \abb{SE} maximization problem with both the data power and \abb{LSFD} vectors as optimization variables and develop an algorithm based on the weighted \abb{MMSE} (minimum mean square error) approach to obtain a stationary point with low computational complexity.
		\end{abstract}
		
	   \IEEEpeerreviewmaketitle
		
		\section{Introduction}

		Massive \abb{MIMO} is an emerging technology to handle the growing demand for wireless data traffic in the next generation cellular networks \cite{Andrews2014a}. A Massive \abb{MIMO} \abb{BS} is equipped with hundreds of antennas to spatially multiplex a large number of users on the same time--frequency resource \cite{marzetta2010noncooperative}. In a single-cell scenario, there is no need for computationally heavy decoding (or precoding methods) in Massive \abb{MIMO} as both the thermal noise and mutual interference are effectively suppressed by linear processing, e.g., maximum-ratio combining (\abb{MRC}) or  regularized zero-forcing (RZF) combining, with a large number of \abb{BS} antennas \cite{Marzetta2016a}. In a multi-cell scenario, however, pilot-based channel estimation is contaminated by the non-orthogonal transmission in other cells.  This results in coherent intercell interference in the data transmission, so-called \emph{pilot contamination} \cite{Jose2011b}, unless high-complexity signal processing schemes are used to suppress it \cite{bjornson2018a}. Pilot contamination reduces the benefit of having many antennas and the \abb{SE}s achieved by low-complexity \abb{MRC} or RZF saturate as the number of antennas grows. 
		
	   Much research has been dedicated to mitigating the effects of pilot contamination; for example by increasing the length of the pilots \cite{Bjornson2016a} or assigning the pilots in a way that reduces the contamination \cite{Jin2015a}. In practical networks, however, it is not possible to make all pilots orthogonal due to the limited channel coherence block \cite{Bjornson2016a}. Besides, the pilot assignment is a combinatorial problem. Even though heuristic algorithms with relatively low complexity can be developed, this approach still suffers from the asymptotic \abb{SE} saturation since we only change one contaminating user for a less contaminating user. Instead of combating pilot contamination, one can utilize decoding schemes where the BSs are cooperating \cite{adhikary2017a}.  In the two-layer \abb{LSFD} (\textit{large-scale fading decoding}) framework, each \abb{BS} applies an arbitrary local linear decoding method in the first layer. The results are then gathered at a common central station that applies so-called \abb{LSFD} vectors in a second-layer to combine the signals from multiple \abb{BS}s to suppress pilot contamination and other inter-cell interference. The \abb{LSFD} vectors are selected only based on the channel statistics (large-scale fading) and, therefore, there is no need for the \abb{BS}s to share their local channel estimates.
 The new decoding design attains high \abb{SE} even with a limited number of \abb{BS} antennas \cite{adhikary2017a}. Previous works on LSFD has either considered uncorrelated Rayleigh fading channels \cite{nayebi2016,adhikary2017a} or special correlated Rayleigh fading based on the one-ring model \cite{Adhikary2018a}. The latter paper optimizes the system with respect to network-wide max-min fairness, which is a criterion that gives all the users the same \abb{SE}, but usually a very low such \abb{SE} \cite{Bjornson2017bo}. 
 
	   In this paper, we generalize the \abb{LSFD} method from \cite{adhikary2017a,Adhikary2018a} to a scenario with arbitrary spatial correlation and also develop a method for joint power control and \abb{LSFD} vector optimization in the system using the sum \abb{SE} as the utility. 
	    We first quantify the  \abb{SE} in a system with arbitrary processing in the two layers and then derive a closed-form expression for the case when \abb{MRC} is used in the first  layer.
	     The \abb{LSFD} vector that maximizes the \abb{SE} follows in closed form. Additionally, an uplink sum \abb{SE} optimization problem with power constraints is formulated. Because it is a hard non-convex problem, we are not searching for the global optimum but develop an alternating low-complexity optimization algorithm that converges to a stationary point. Numerical results demonstrate the effectiveness of the optimized system for Massive \abb{MIMO} systems with correlated Rayleigh fading.
		
		\textit{Notation}: Lower and upper case bold letters are used for vectors and matrices.  The expectation of a random variable $X$ is denoted by $\Exp \{X \}$ and  $\| \cdot \|$ is the Euclidean norm. The transpose and Hermitian transpose of a matrix are written as $(\cdot)^\tr$ and $(\cdot)^\conjtr$, respectively. The $L\times L$\mbox{-}dimensional diagonal matrix with the diagonal elements $d_1, d_2,\ldots,d_L$ is denoted $\operatorname{diag}(d_1,d_2,\ldots,d_L)$.  Finally, $\mathcal{CN}(\cdot, \cdot)$ is  circularly symmetric complex Gaussian distribution. 
		
		\section{System Model} \label{Section: System Model}
		
		We consider a cellular network with $L$~cells. Each cell consists of a \abb{BS} equipped with $M$ antennas that serves $K$ single-antenna users. The channel vector in the uplink between user~$k$ in cell~$l$ and \abb{BS}~$l'$ is denoted by $\mathbf{h}_{l,k}^{l'} \in \mathbb{C}^M$. We consider the standard block-fading model \cite{Bjornson2017bo}, where the channels are static within a coherence block of size $\tau_\mathrm{c}$ channel uses and take an independent realization in each blocks, 
		according to a stationary ergodic random process. Since practical channels are spatially correlated, we assume that each channel follows a correlated Rayleigh fading model:
		\begin{equation}
		\mathbf{h}_{l,k}^{l'} \sim \mathcal{CN} \left( \mathbf{0}, \mathbf{R}_{l,k}^{l'} \right),
		\end{equation}
		where $\mathbf{R}_{l,k}^{l'} \in \mathbb{C}^{M \times M}$ is the \emph{spatial correlation matrix}. 
		The \abb{BS}s know the channel statistics, 
		but have no prior knowledge of the channel realizations, which need to be estimated in every coherence block.

		\subsection{Channel Estimation}
		As in conventional Massive \abb{MIMO} \cite{bjornson2018a},  
		the channels are estimated by letting the users transmit $K$-symbol long \emph{pilots} in a dedicated part of the coherence block, called the \emph{pilot phase}. All the cells share a common set of $K$ mutually orthogonal pilots $\{\pmb{\phi}_1, \ldots, \pmb{\phi}_K \}$ with $\| \pmb{\phi}_k \|^2 = K$.	Without loss of generality, we assume that  the users in different cells that have the same index use the same pilot and thereby cause pilot contamination to each other \cite{marzetta2010noncooperative}. During the pilot phase, at BS~$l$  the received signal in the pilot phase is denoted $\mathbf{Y}_l \in \mathbb{C}^{M\times K}$ and it is given by
		\begin{equation}
		\mathbf{Y}_l = \sum_{l'=1}^L \sum_{k=1}^K \sqrt{\hat{p}_{l',k}}  \mathbf{h}_{l',k}^l \boldsymbol{\phi}_{k}^\conjtr + \mathbf{N}_l,
		\end{equation}
		where $\hat{p}_{l',k}$ is the power of the pilot transmitted by user~$k$ in cell~$l'$ and $\mathbf{N}_l$ is a matrix of independent and identically distributed noise terms, each distributed as $\mathcal{CN}(0,\sigma^2)$. An observation of the channel from user~$k$ to BS~$l$ is obtained by using standard \abb{MMSE} estimation \cite{Bjornson2017bo}. 
		The channel estimates are used at BS~$l$ to compute decoding vectors for detecting the signals from the $K$ intra-cell users. 
		
		\subsection{Uplink Data Transmission}
		
		In the \emph{data phase}, it is assumed that user~$k$ in cell~$l'$ sends a zero-mean information symbol~$s_{l',k}$ with power $\Exp \{ |s_{l',k}|^2 \} = 1$. 
		The received signal $\mathbf{y}_l \in \mathbb{C}^M $ at BS~$l$ is then
		\begin{equation} \label{eq:ReceivedSigBSl}
		\mathbf{y}_l = \sum_{l'=1}^L \sum_{k=1}^K \sqrt{p_{l',k}} \mathbf{h}_{l',k}^l s_{l',k} + \mathbf{n}_l,
		\end{equation}
		where $p_{l',k}$ denotes the transmit power of user~$k$ in cell~$l'$. 
		Based on the signals in \eqref{eq:ReceivedSigBSl}, the BSs decode the symbols using the two-layers-decoding technique that is illustrated in Fig.~\ref{fig:decoder}. 
		In the first layer, an estimate of the symbol from user~$k$ in cell~$l$ is obtained at BS~$l$ by local linear decoding as
		\begin{equation}
		\tilde{s}_{l,k} = \mathbf{v}_{l,k}^\conjtr \mathbf{y}_l = \sum_{l'=1}^L \sum_{k'=1}^K \sqrt{p_{l',k'}} \mathbf{v}_{l,k}^\conjtr \mathbf{h}_{l',k'}^l s_{l',k'} + \mathbf{v}_{l,k}^\conjtr \mathbf{n}_l,
		\end{equation} 
		where $\mathbf{v}_{l,k}$ is the \emph{linear decoding vector}. The symbol estimate $\tilde{s}_{l,k}$ contains interference and noise. In particular, the coherent interference caused by pilot contamination from pilot-sharing users in other cells is large in  Massive \abb{MIMO}. To mitigate the inter-cell interference, all the symbol estimates of the pilot-sharing users are collected in a vector 
		\begin{equation}
		\tilde{\mathbf{s}}_{k} \defas [\tilde{s}_{1,k}, \tilde{s}_{2,k}, \ldots,  \tilde{s}_{L,k} ]^\tr \in \mathbb{C}^L.
		\end{equation}
		After the local decoding, a second layer of centralized decoding is performed. 
		The final estimate of the data symbol from user~$k$ in cell~$l$ is obtained as		\begin{equation} \label{eq:919838377}
		\begin{split}
		& \hat{s}_{l,k} = \mathbf{a}_{l,k}^\conjtr \tilde{\mathbf{s}}_{k} = \sum_{l'=1}^L (a_{l,k}^{l'})^{\ast} \tilde{s}_{l',k}
		\end{split}
		\end{equation}
		where $\mathbf{a}_{l,k} \defas [a_{l,k}^1, a_{l,k}^2, \ldots, a_{l,k}^L]^\tr \in \mathbb{C}^L$ is called the \emph{\abb{LSFD} vector} and $a_{l,k}^{l'}$ is the \emph{\abb{LSFD} weight}. {Unlike previous works, in our framework, arbitrary linear combining methods can be used in the first layer and the \abb{LSFD} vectors can still be optimized.} 
		
		In the next section, we use the decoded signals $\hat{s}_{l,k}$ together with the asymptotic channel properties \cite[Section~2.5]{Bjornson2017bo} to derive a closed-from expression for achievable uplink \abb{SE}.
		\begin{figure}[t]
			\centering
			\includegraphics[trim=0.0cm 0cm 0.0cm 0cm, clip=true, width=2.3in]{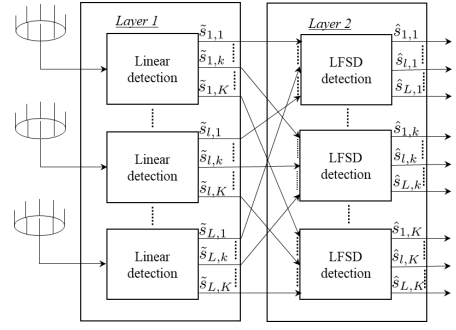} \vspace*{-0.15cm}
			\caption{Desired signals are detected by the two-layer decoding technique.}
			\label{fig:decoder}
			\vspace*{-0.5cm}
		\end{figure}
		\section{Uplink Performance Analysis} \label{Section:ULPerformance}
		This section first derives a  \abb{SE} expression that can be used for any decoding vector and then a closed-form expression when using \abb{MRC}. 
		These expressions are then used to obtain the \abb{LSFD} vectors that maximize the \abb{SE}. 
		The use-and-then-forget capacity bounding technique \cite[Chapter~2.3.4]{Marzetta2016a}, \cite[Section~4.3]{bjornson2018a} allows us to compute a lower bound on the uplink ergodic capacity (i.e., an achievable \abb{SE}) of user~$k$ in cell~$l$ as
			\begin{equation} \label{eq:GeneralSE}
			R_{l,k} = \max_{\{a_{l,k}^{l'}\}} \left( 1- \frac{K}{\tau_\mathrm{c}} \right) \log_2 \left(1 + \mathrm{SINR}_{l,k} \right),
			\end{equation}
			where the effective SINR, denoted by $\mathrm{SINR}_{l,k}$, is
			\begin{equation} \label{eq:GeneralSINR}
			 \frac{ \Exp\{ |\mathtt{DS}_{l,k}|^2 \} }{ \Exp \{ |\mathtt{PC}_{l,k} |^2 \} + \Exp \{ |\mathtt{BU}_{l,k}|^2\} +  \Exp \{ |\mathtt{NI}_{l,k}|^2\} + \Exp \{ |\mathtt{AN}_{l,k}|^2\}},
			\end{equation}
			where $\mathtt{DS}_{l,k}, \mathtt{PC}_{l,k}, \mathtt{BU}_{l,k},\mathtt{NI}_{l,k},$ and $\mathtt{AN}_{l,k}$ stand for the desired signal, the pilot contamination, the beamforming gain uncertainty, the non-coherent interference, and the additive noise, respectively, whose expectations are defined as
			\begin{align}
			\mathbb{E} \{ |\mathtt{DS}_{l,k }|^2 \} &\defas p_{l,k} \left| \sum_{l'=1}^L (a_{l,k}^{l'})^{\ast} \Exp \{\mathbf{v}_{l',k}^\conjtr \mathbf{h}_{l,k}^{l'} \}\right|^2 ,\\
			\mathbb{E} \{ |\mathtt{PC}_{l,k}|^2 \} &\defas  \sum_{\substack{l''=1 \\ l'' \neq l }}^L p_{l'',k} \left| \sum_{l'=1}^L (a_{l,k}^{l'})^{\ast}  \Exp \{\mathbf{v}_{l',k}^\conjtr \mathbf{h}_{l'',k}^{l'} \} \right|^2,\\
			\mathbb{E} \{|\mathtt{BU}_{l,k}|^2 \} &\defas \sum_{l'=1}^L p_{l',k} \Exp \Bigg\{ \Bigg| \sum_{l''=1}^L (a_{l,k}^{l''})^{\ast}  \Bigg( \mathbf{v}_{l'',k}^\conjtr \mathbf{h}_{l',k}^{l''} - \notag \\ & \quad \Exp \{\mathbf{v}_{l'',k}^\conjtr \mathbf{h}_{l',k}^{l''} \} \Bigg) \Bigg|^2 \Bigg\},\\
			\mathbb{E} \{ |\mathtt{NI}_{l',k'}|^2 \} &\defas \sum\limits_{l'=1}^{L} \sum\limits_{\substack{k'=1\\k'\neq k}}^K  p_{l',k'} \times \notag \\
			& \quad \Exp \left\{ \left| \sum_{l''=1}^L (a_{l,k}^{l''})^{\ast}  \mathbf{v}_{l'',k}^\conjtr \mathbf{h}_{l',k'}^{l''} \right|^2 \right\} , \\
			\mathbb{E} \{ |\mathtt{AN}_{l,k}|^2 \} &\defas \Exp \left\{ \left| \sum_{l'=1}^L (a_{l,k}^{l'})^\ast  (\hat{\mathbf{h}}_{l',k}^{l'})^\conjtr \mathbf{n}_{l'} \right|^2 \right\}.
			\end{align} 
		We notice that the SE expression in \eqref{eq:GeneralSE} can be applied together with any linear decoding method and any \abb{LSFD} vector, but the expectations have the evaluated numerically.
		
		Maximizing the \abb{SE} of user~$k$ in cell~$l$ is equivalent to selecting the \abb{LSFD} vector that maximizes a \textit{Rayleigh quotient}.
		\begin{theorem} \label{Theorem1v1}
			If \abb{MRC}, \abb{ZF} or \abb{RZF} is used, for a given set of pilot and data power coefficients, the \abb{SE} of user~$k$ in cell~$l$ is
			\begin{equation} \label{eq:OptimalRate}
			R_{l,k} = \left( 1- \frac{K}{\tau_\mathrm{c}} \right) \log_2 \left(1 +  p_{l,k} \mathbf{b}_{l,k}^\conjtr \left( \sum_{i=1}^4 \mathbf{C}_{l,k}^{(i)}  \right)^{-1} \mathbf{b}_{l,k} \right),
			\end{equation}
			where the matrices $\mathbf{C}_{l,k}^{(1)},\mathbf{C}_{l,k}^{(2)}, \mathbf{C}_{l,k}^{(3)}, \mathbf{C}_{l,k}^{(4)} \in \mathbb{C}^{L \times L}$ are 
			\begin{align}
			\mathbf{C}_{l,k}^{(1)} &\defas \sum_{\substack{l'=1 \\ l' \neq l }}^L p_{l',k} \mathbf{b}_{l',k} \mathbf{b}_{l',k}^\conjtr, \\
			\mathbf{C}_{l,k}^{(2)} &\defas \sum_{l'=1}^L p_{l',k} \Exp \left\{ \tilde{\mathbf{b}}_{l',k}  \tilde{\mathbf{b}}_{l',k}^\conjtr \right\},\\
			\mathbf{C}_{l,k}^{(3)} &\defas \mathrm{diag} \left(\sum_{l'=1}^L \sum_{\substack{k'=1 \\ k' \neq k}}^K  p_{l',k'} \Exp \left\{\left|  \mathbf{v}_{1,k}^\conjtr \mathbf{h}_{l',k'}^{1} \right|^2  \right\}, \ldots, \right. \notag \\& \left. \sum_{l'=1}^L \sum_{\substack{k'=1 \\ k' \neq k}}^K  p_{l',k'} \Exp \left\{\left|  \mathbf{v}_{L,k}^\conjtr \mathbf{h}_{l',k'}^{L} \right|^2  \right\} \right) \label{eq:C3}\\
			\mathbf{C}_{l,k}^{(4)} &\defas \mathrm{diag} \left( \sigma^2 \Exp \left\{ \| \mathbf{v}_{1,k} \|^2 \right\}, \ldots, \sigma^2  \Exp  \left\{ \| \mathbf{v}_{L,k} \|^2 \right\} \right),
			\end{align}
			and the vectors $\mathbf{b}_{l',k}, \tilde{\mathbf{b}}_{l',k} \in \mathbb{C}^L, \forall l' = 1,\ldots,L,$ are
			\begin{align}
			\mathbf{b}_{l',k} &\defas \left[ 
			\Exp  \{ \mathbf{v}_{1,k}^H \mathbf{h}_{l',k}^1 \}, \ldots, \Exp  \{ \mathbf{v}_{L,k}^H \mathbf{h}_{l',k}^L \} \right]^\tr, \label{eq:bikv1}\\
			\tilde{\mathbf{b}}_{l',k} &\defas \left[ \mathbf{v}_{1,k}^\conjtr \mathbf{h}_{l',k}^{1},  \ldots, \mathbf{v}_{L,k}^\conjtr \mathbf{h}_{l',k}^{L}\right]^\tr - \mathbf{b}_{l',k}.
			\end{align}
			In order to attain this \abb{SE}, the \abb{LSFD} vector is selected as
			\begin{equation} \label{eq:LSFDVector}
			\mathbf{a}_{l,k} = \left( \sum_{i=1}^4 \mathbf{C}_{l,k}^{(i)} \right)^{-1} \mathbf{b}_{l,k}, \quad\forall l,k.
			\end{equation}
		\end{theorem}
		\begin{IEEEproof}
			The proof relies on rewriting the SE as a generalized Rayleigh quotient and solving it. The details are available in the journal version of this paper \cite{Chien2018large}.
		\end{IEEEproof}
		We stress that the \abb{LSFD} vector in \eqref{eq:LSFDVector} is designed to maximize the \abb{SE} in \eqref{eq:OptimalRate} for every user in the network for a given data and pilot power and a given first-layer decoding method. 
		This is a non-trivial generalization of the previous works \cite{nayebi2016,adhikary2017a,Adhikary2018a}, which only considered specific first-layer decoding methods that could provide closed-form expressions.\footnote{We stress that Theorem~\ref{Theorem1v1} also holds in other cases, if we replace $\mathbf{C}_{l,k}^{(3)}$ as $\mathbf{C}_{l,k}^{(3)} = \sum_{l'=1}^L \sum_{k'=1, k' \neq k}^K p_{l',k'} \Exp  \{ \mathbf{z}_{k,l',k'} \mathbf{z}_{k,l',k'}^H \}$, where $\mathbf{z}_{k,l',k'} = [ \mathbf{v}_{1,k}^H \mathbf{h}_{l',k'}^1, \ldots, \mathbf{v}_{L,k}^H \mathbf{h}_{l',k'}^L ] \in \mathbb{C}^L$.} The following theorem states a closed-form expression of the \abb{SE} for the case of \abb{MRC} in arbitrary spatial correlation, which makes the results more practical than in \cite{Adhikary2018a}. 
		\begin{theorem} \label{Theorem2}
			When \abb{MRC} is used, the \abb{SE} in \eqref{eq:GeneralSE} of user~$k$ in cell~$l$ is given by
			\begin{equation} \label{eq: ClosedForm_Rate_MMSE}
			R_{l,k} = \left( 1- \frac{K}{\tau_c} \right) \log_2 \left(1 + \mathrm{SINR}_{l,k} \right),
			\end{equation}
			where the SINR value is given in \eqref{eq:MMSE_SINR} on the top of the next page.
			\begin{figure*}
				\begin{equation} \label{eq:MMSE_SINR}
				\mathrm{SINR}_{l,k} = \frac{  p_{l,k} \left| \sum_{l'=1}^L (a_{l,k}^{l'})^{\ast} b_{l,k}^{l'} \right|^2  }{  \sum_{l'=1, l' \neq l }^L p_{l',k} \left| \sum_{l''=1}^L (a_{l,k}^{l''})^{\ast} b_{l',k}^{l''} \right|^2 + \sum_{l'=1}^L \sum_{k'=1}^K  \sum_{l''=1}^L  p_{l',k'} |a_{l,k}^{l''}|^2 c_{l'',k}^{l',k'}  + \sum_{l'=1}^L |a_{l,k}^{l'}|^2 d_{l',k} }
				\end{equation}
				\vspace*{-0.6cm}
				\hrulefill
			\end{figure*}
			The values $b_{l',k}^{l''}, c_{l'',k}^{l',k'},$ and $d_{l',k}$ are given as
			\begin{align}
			b_{l',k}^{l''} &= \sqrt{K \hat{p}_{l',k} \hat{p}_{l'',k}} \mathrm{tr} \left( \pmb{\Psi}_{l'',k}^{-1} \mathbf{R}_{l'',k}^{l''} \mathbf{R}_{l',k}^{l''} \right),\label{eq:8926616}\\
			c_{l'',k}^{l',k'} &= \hat{p}_{l'',k} \mathrm{tr} \left( \mathbf{R}_{l'',k}^{l''} \pmb{\Psi}_{l'',k}^{-1} \mathbf{R}_{l'',k}^{l''} \mathbf{R}_{l',k'}^{l''} \right),\\
			d_{l',k} &= \sigma^2 \hat{p}_{l',k} \mathrm{tr} \left( \pmb{\Psi}_{l',k}^{-1} \mathbf{R}_{l',k}^{l'} \mathbf{R}_{l',k}^{l'} \right),
			\end{align}
			where $\pmb{\Psi}_{l'',k} = K \sum_{l=1}^K \hat{p}_{l,k} \mathbf{R}_{l,k}^{l''} + \sigma^2 \mathbf{I}_M$ and $\pmb{\Psi}_{l',k}$ is defined in the same manner.
		\end{theorem}
		\begin{proof}
			The proof encompasses of computing the moments of complex Gaussian distributions and the detail is available in the journal version of this paper \cite{Chien2018large}.
		\end{proof}
		
		Theorem~\ref{Theorem2} describes the exact impact that the spatial correlation has on the system performance through the coefficients $b_{l',k}^{l''}, c_{l'',k}^{l',k'},$ and $d_{l',k}$.  
		It is seen that the numerator of \eqref{eq:MMSE_SINR} grows as the square of the number of antennas, $M^2$, since the trace in \eqref{eq:8926616} is the sum of $M$ terms. This gain comes from the coherent combination of the signals from the $M$ antennas.  It can also be seen from Theorem~\ref{Theorem2} that the pilot contamination in \eqref{eq:919838377} combines coherently, i.e., its variance---the first term in the denominator that contains $b^{l''}_{l,k}$---grows as $M^2$.  The other terms in the denominator represent the impact of non-coherent interference and Gaussian noise, respectively.  These two terms only grow as $M$. Since the interference terms contain products of correlation matrices of different users, the interference is smaller between users that have different spatial correlation characteristics \cite{Bjornson2017bo}.
		
		The following corollary gives the optimal \abb{LSFD} vector $\mathbf{a}_{l,k}$ that maximizes the \abb{SE} of every user for a given set of pilot and data powers. 
		\begin{corollary} \label{corollary:Opt_LSFD}
			For a given set of data and pilot powers, by using \abb{MRC} and optimal \abb{LSFD}, the \abb{SE} in Theorem~\ref{Theorem2} is given in closed form as
			\begin{equation} \label{eq:MaxRateMRC}
			R_{l,k} = \left( 1- \frac{K}{\tau_\mathrm{c}} \right) \log_2 \left(1 +  p_{l,k}\mathbf{b}_{l,k}^\conjtr  \mathbf{C}_{l,k}^{-1} \mathbf{b}_{l,k} \right)
			\end{equation}
			where $\mathbf{C}_{l,k}\in \mathbb{C}^{L \times L}$ and $\mathbf{b}_{l,k} \in \mathbb{C}^{L}$ are defined as
			\begin{align}
			\mathbf{C}_{l,k} &\defas \sum_{ \substack{l'=1\\l' \neq l}}^L  p_{l',k} \mathbf{b}_{l',k} \mathbf{b}_{l',k}^\conjtr + \mathrm{diag} \left( \sum_{l'=1}^L \sum_{k'=1}^K p_{l',k'} c_{1,k}^{l',k'} + d_{1,k} \right.  \notag \\ & \left. ,\ldots, \sum_{l'=1}^L \sum_{k'=1}^K p_{l',k'} c_{L,k}^{l',k'} + d_{L,k} \right), \label{eq:Clk}\\
			\mathbf{b}_{l',k} &\defas [b_{l',k}^1, \ldots, b_{l',k}^L]^\tr. \label{eq:bik}
			\end{align}
			The \abb{SE} in \eqref{eq:MaxRateMRC} is obtained by using \abb{LSFD} vector 
			\begin{equation}
			\mathbf{a}_{l,k} = \mathbf{C}_{l,k}^{-1} \mathbf{b}_{l,k}.
			\end{equation}
		\end{corollary}
		Although Corollary~\ref{corollary:Opt_LSFD} is a special case of Theorem~\ref{Theorem1v1} when \abb{MRC} is used, the \abb{LSFD} vector $\mathbf{a}_{l,k}$ is  obtained in closed form.
		
		\section{Optimizing the Sum \abb{SE} } \label{Section:SumRateOpt}
		In this section, the sum \abb{SE} maximization problem is formulated where the 
		optimization variables are the data powers and \abb{LSFD} vectors. Since this problem is NP-hard, an iterative algorithm is proposed to  find a stationary point with low computational complexity.
		
		\subsection{Problem Formulation}
		We consider sum \abb{SE} maximization
		\begin{equation} \label{Problem: Sumrate}
		\begin{aligned} 
		 \underset{ \{ p_{l,k} \geq 0 \}, \{ \mathbf{a}_{l,k} \} }{\textrm{maximize}}
		&&&  \sum_{l=1}^L \sum_{k=1}^K R_{l,k} \\
		 \mbox{subject to} \quad
		 &&& p_{l,k} \leq P_{\mathrm{max},l,k} \quad \forall l,k.
		\end{aligned}
		\end{equation}
		Inserting the SE expression \eqref{eq: ClosedForm_Rate_MMSE} into \eqref{Problem: Sumrate}, and removing the constant pre-log factor, yields the equivalent formulation
		\begin{equation} \label{Problem: Sumratev1}
		\begin{aligned}
		 \underset{ \{ p_{l,k} \geq 0 \},  \{ \mathbf{a}_{l,k} \}  }{\textrm{maximize}}
		&&&  \sum_{l=1}^L \sum_{k=1}^K \log_2 \left(1 + \mathrm{SINR}_{l,k} \right) \\
		 \mbox{subject to} \quad
		&&& p_{l,k} \leq P_{\mathrm{max},l,k} \quad \forall l,k.
		\end{aligned}
		\end{equation}
		Sum \abb{SE} maximization with imperfect \abb{CSI} is known to be a non-convex and \abb{NP}-hard problem \cite{Annapureddy2011a} and this applies also to \eqref{Problem: Sumratev1}, even if the optimal LSFD vectors are given in Corollary~\ref{corollary:Opt_LSFD}.
		Therefore, the global optimum is overly difficult to compute. 
		Nevertheless, solving the ergodic sum SE maximization in~\eqref{Problem: Sumratev1} for a Massive \abb{MIMO} system is more practical than maximizing the instantaneous \abb{SE}s for a given small-scale fading realization, as is normally done in small-scale \abb{MIMO} systems \cite{Christensen2008a}. 
		Since the sum \abb{SE} maximization in \eqref{Problem: Sumratev1} only depends on the large-scale fading coefficients, the solution can be used for as much time  as the channel statistics are constant.
		Another key difference from prior work is that we jointly optimize the data powers and \abb{LSFD} vectors.
		
		Instead of seeking the global optimum to~\eqref{Problem: Sumratev1}, we will obtain a stationary point to \eqref{Problem: Sumratev1} by following the weighted \abb{MMSE} approach from \cite{Christensen2008a} and adapt it to the problem at hand. 
		To this end, we first formulate the weighted \abb{MMSE} problem that is equivalent to \eqref{Problem: Sumratev1}.
		\begin{theorem} \label{Theorem:MMSEOptProblem}
			The optimization problem
			\begin{equation} \label{Problem: Sumratev2}
			\begin{aligned}
			 \underset{ \substack{\{ p_{l,k} \geq 0 \}, \{ \mathbf{a}_{l,k} \},\\ \{ w_{l,k} \geq 0 \}, \{ u_{l,k} \} }}{\mathrm{minimize}}
			&&&  \sum_{l=1}^L \sum_{k=1}^K w_{l,k} e_{l,k} - \ln (w_{l,k}) \\
			 \mathrm{subject \, \, to} \quad
			&&& p_{l,k} \leq P_{\mathrm{max},l,k} \;, \forall l,k,\\
			\end{aligned}
			\end{equation}
			where $e_{l,k}$ is defined as
			\begin{equation} \label{eq: errorOpt}
			\begin{split}
			e_{l,k} &\defas  |u_{l,k}|^2 \left( \sum\limits_{l'=1 }^L p_{l',k} \left| \sum\limits_{l''=1}^L  (a_{l,k}^{l''})^{\ast} b_{l',k}^{l''} \right|^2  +  \sum\limits_{l'=1}^L \sum\limits_{k'=1}^K \sum\limits_{l''=1}^L  \right. \\ &\left. p_{l',k'}| a_{l,k}^{l''} |^2 c_{l'',k}^{l',k'} +  \sum\limits_{l'=1}^L | a_{l,k}^{l'}|^2 d_{l',k}  \right) - 2\sqrt{p_{l,k}}  \times\\
			&\mathfrak{Re}\left(u_{l,k} \left( \sum_{l'=1}^L (a_{l,k}^{l'})^{\ast} b_{l,k}^{l'} \right)\right)   +1,
			\end{split}
			\end{equation}
			is equivalent to the sum \abb{SE} optimization problem \eqref{Problem: Sumratev1} in the sense that \eqref{Problem: Sumratev1} and \eqref{Problem: Sumratev2} have the same optimal transmit powers $\{p_{l,k}\}, \forall l,k,$ and \abb{LSFD} vectors $ \{ \mathbf{a}^{l'}_{l,k} \}, \forall l,k,l'$. 
		\end{theorem}
		\begin{proof}
			The proof is based on the signal detection process of a SISO system having the same SE as \eqref{eq:MaxRateMRC}. The detail proof is available in our journal version \cite{Chien2018large}.
		\end{proof}
		
		\setcounter{eqnback}{\value{equation}} \setcounter{equation}{37}
		\begin{figure*}
			\begin{equation} \label{eq:tildeu_lk_n1}
			\tilde{u}_{l,k}^{(n-1)} = \sum\limits_{l'=1}^L (\rho_{l',k}^{(n-1)})^2 \left| \sum\limits_{l''=1}^L (a_{l,k}^{l'', (n-1)})^{\ast} b_{l',k}^{l''} \right|^2  + \sum\limits_{l'=1}^L \sum\limits_{k'=1}^K \sum\limits_{l''=1}^L (\rho_{l',k'}^{(n-1)})^2 | a_{l,k}^{l'',(n-1)}|^2 c_{l'',k}^{l',k'}
			+  \sum\limits_{l'=1}^L | a_{l,k}^{l',(n-1)}|^2 d_{l',k}
			\end{equation} \vspace*{-0.65cm}
		\end{figure*}
		\setcounter{eqncnt}{\value{equation}}
		\setcounter{equation}{\value{eqnback}}
		
		\setcounter{eqnback}{\value{equation}} \setcounter{equation}{41}
		\begin{figure*}
			\begin{equation} \label{eq:C_lk_n}
			\widetilde{\mathbf{C}}_{l,k}^{(n-1)} = \sum_{l'=1}^L  (\rho_{l',k}^{(n-1)})^2 \mathbf{b}_{l',k} \mathbf{b}_{l',k}^\conjtr + \mathrm{diag} \left( \sum_{l'=1}^L \sum_{k'=1}^K (\rho_{l',k'}^{(n-1)})^2 c_{1,k}^{l',k'} + d_{1,k}, \ldots, \sum_{l'=1}^L \sum_{k'=1}^K (\rho_{l',k'}^{(n-1)})^2 c_{L,k}^{l',k'} + d_{L,k} \right)
			\end{equation} 
		    \vspace*{-0.25cm}
			\begin{equation} \label{eq:rho_lkSol}
			\rho_{l,k}^{(n)} = \min \left(\frac{w_{l,k}^{(n)} \mathfrak{Re}  \left( u_{l,k}^{(n)}   \sum_{l'=1}^L (a_{l,k}^{l',(n)})^{\ast} b_{l,k}^{l'} \right) }{\sum_{l' =1}^L w_{l',k}^{(n)} |u_{l',k}^{(n)}|^2 \left| \sum_{l''=1}^L (a_{l',k}^{l'',(n)})^{\ast} b_{l,k}^{l''} \right|^2 +  \sum_{l'=1}^L \sum_{k'=1}^K w_{l',k'}^{(n)} |u_{l',k'}^{(n)}|^2 \sum_{l''=1}^L | a_{l',k'}^{l'',(n)}|^2 c_{l'',k'}^{l,k}}, \sqrt{P_{\max,l,k}} \right)
			\end{equation} \vspace*{-0.4cm}
			\hrulefill
			\vspace*{-0.2cm}
		\end{figure*}
		\setcounter{eqncnt}{\value{equation}}
		\setcounter{equation}{\value{eqnback}}
		\vspace*{-0.2cm}
		\subsection{Iterative Algorithm}
		We will obtain a stationary point to \eqref{Problem: Sumratev2} by decomposing it into a sequence of subproblems, each having a closed-form solution.
		To this end, the power variable $p_{l,k}$ is substituted with $\rho_{l,k} = \sqrt{p_{l,k}}$. By alternating between solving the subproblems we obtain the following result. 
		\begin{theorem} \label{Theorem: IterativeSol}
		 A stationary point to \eqref{Problem: Sumratev2} is obtained by iteratively updating $\{ \mathbf{a}_{l,k}, u_{l,k}, w_{l,k}, \rho_{l,k} \}$. Let $ \mathbf{a}_{l,k}^{n-1}, u_{l,k}^{n-1}, w_{l,k}^{n-1}, \rho_{l,k}^{n-1}$ the values after iteration $n-1$. At iteration $n$, the optimization parameters are updated in the following way:
		 \begin{itemize}[leftmargin=*]
			\item $u_{l,k}$ is updated as
				\begin{equation} \label{eq:u_lk_n1}
				u_{l,k}^{(n)} = \left(\rho_{l,k}^{(n-1)}  \sum\limits_{l'=1}^L a_{l,k}^{l',(n-1)} (b_{l,k}^{l'})^{\ast} \right)/ \tilde{u}_{l,k}^{(n-1)},
				\end{equation}
				where the value $\tilde{u}_{l,k}^{(n-1)}$ is defined in \eqref{eq:tildeu_lk_n1} on the top of this page.
				\item $w_{l,k}$ is updated as
				\setcounter{eqnback}{\value{equation}} \setcounter{equation}{38}
				\begin{equation} \label{eq:w_lk_n1}
				w_{l,k}^{(n)} = 1/e_{l,k}^{(n)},
				\end{equation}
				where $e_{l,k}^{(n)}$ is given by
				\begin{multline} \label{eq:e_lk_n1}
				e_{l,k}^{(n)} = |u_{l,k}^{(n)}|^2 \tilde{u}_{l,k}^{(n-1)} - 2\rho_{l,k}^{(n-1)} \times \\
				\mathfrak{Re} \left(u_{l,k}^{(n)}   \left( \sum_{l'=1}^L (a_{l,k}^{l',(n-1)})^{\ast} b_{l,k}^{l'} \right) \right) + 1.
				\end{multline}
				\item $\mathbf{a}_{l,k}$ is updated as
				\begin{equation} \label{eq: a_lk_n1}
				\mathbf{a}_{l,k}^{(n)} = \tilde{u}_{l,k}^{\ast, (n)} \left(\widetilde{\mathbf{C}}_{l,k}^{(n-1)}\right)^{-1} \mathbf{b}_{l,k} / \sum\limits_{l'=1}^L (a_{l,k}^{l',(n-1)})^{\ast} b_{l,k}^{l'},
				\end{equation}
				where $\widetilde{\mathbf{C}}_{l,k}^{(n-1)}$ is computed as in \eqref{eq:C_lk_n} on the top of this page.
				\item $\rho_{l,k}$ is updated as in \eqref{eq:rho_lkSol} on the top of this page.
				
			\end{itemize}
			If we denote the stationary point to \eqref{Problem: Sumratev2} that is attained by the above iterative algorithm as $n\to \infty$ by $u_{l,k}^{\mathrm{opt}}$, $w_{l,k}^{\mathrm{opt}}$, $\mathbf{a}_{l,k}^{\mathrm{opt}}$, and $(\rho_{l,k}^{\mathrm{opt}})^2$, for all $l,k$, 
			then the solution $\{\mathbf{a}_{l,k}^{\mathrm{opt}}\}, \{(\rho_{l,k}^{\mathrm{opt}})^2\},$ is also a stationary point to the problem \eqref{Problem: Sumratev1}.
		\end{theorem}
		\begin{proof}
			The closed-form expression to each optimization variable is obtained by taking the first derivative of the Lagrangian function and equating it to zero, while the same stationary point of problems~\eqref{Problem: Sumratev1} and \eqref{Problem: Sumratev2} is based on the chain rule. The detail proof is available in the journal version \cite{Chien2018large}.
		\end{proof}
		If the initial data power values are uniformly distributed over the range $[0, \sqrt{P_{\max,l,k}}]$, the initial \abb{LSFD} vectors can be computed using Corollary~\ref{corollary:Opt_LSFD}. The iterative algorithm in Theorem~\ref{Theorem: IterativeSol} is then used to obtain a stationary point to problem~\eqref{Problem: Sumrate}. This algorithm is terminated when the variation between two consecutive iterations is sufficiently small.
		
		\section{Numerical Results} \label{Section:NumericalResults}
		
			\begin{figure*}[t]
			\begin{minipage}{0.32\textwidth}
				\centering
				\includegraphics[trim=0.25cm 0.0cm 1.0cm 0.5cm, clip=true, width=2.3in]{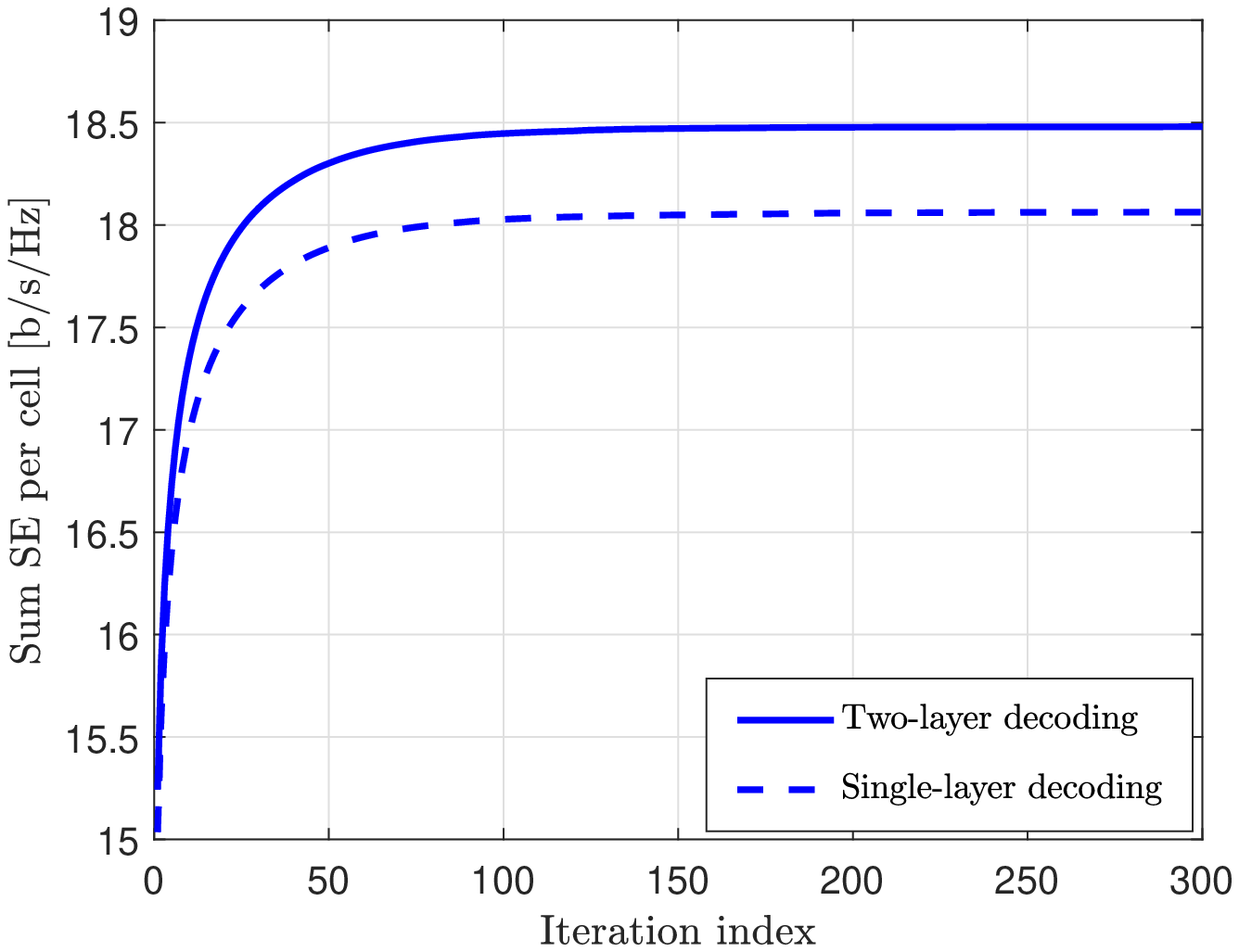} \vspace*{-0.55cm}
				\caption{Convergence of the proposed sum \abb{SE} optimization with $M=200$, $K= 5$, $\varsigma = 0.8$.}
				\label{Fig-Convergence}
				\vspace*{-0.3cm}
			\end{minipage}
			\hfill
			\begin{minipage}{0.32\textwidth}
			\centering
			\includegraphics[trim=0.25cm 0.0cm 1.0cm 0.5cm, clip=true, width=2.3in]{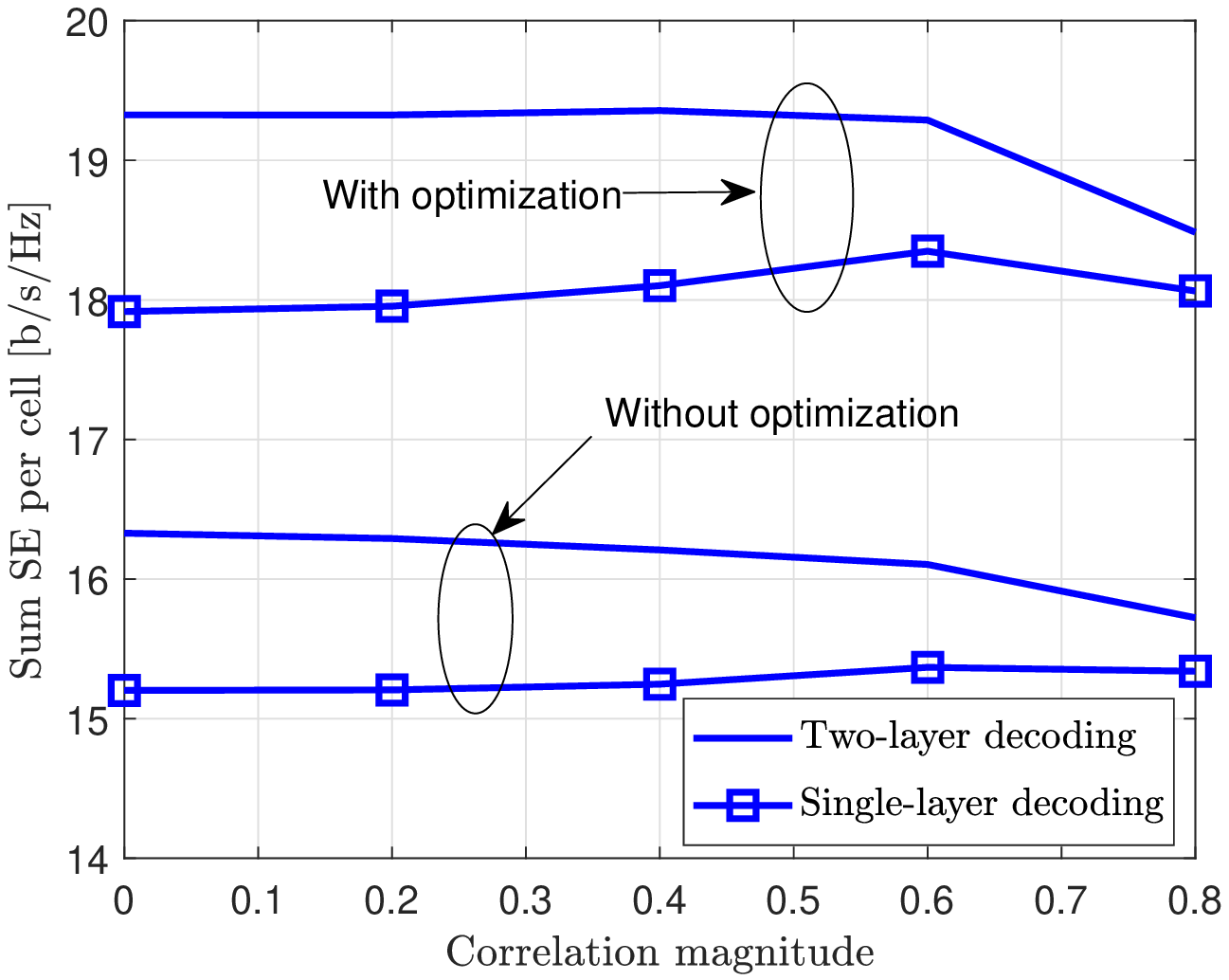} \vspace*{-0.55cm}
			\caption{Sum \abb{SE} per cell versus different correlation magnitudes with $M=200$, $K= 5$.}
			\label{Fig-MMSE-DiffCofact}
			\vspace*{-0.3cm}
		\end{minipage}
		\hfill
		\begin{minipage}{0.32\textwidth}
			\centering
			\includegraphics[trim=0.25cm 0.0cm 1.0cm 0.5cm, clip=true, width=2.3in]{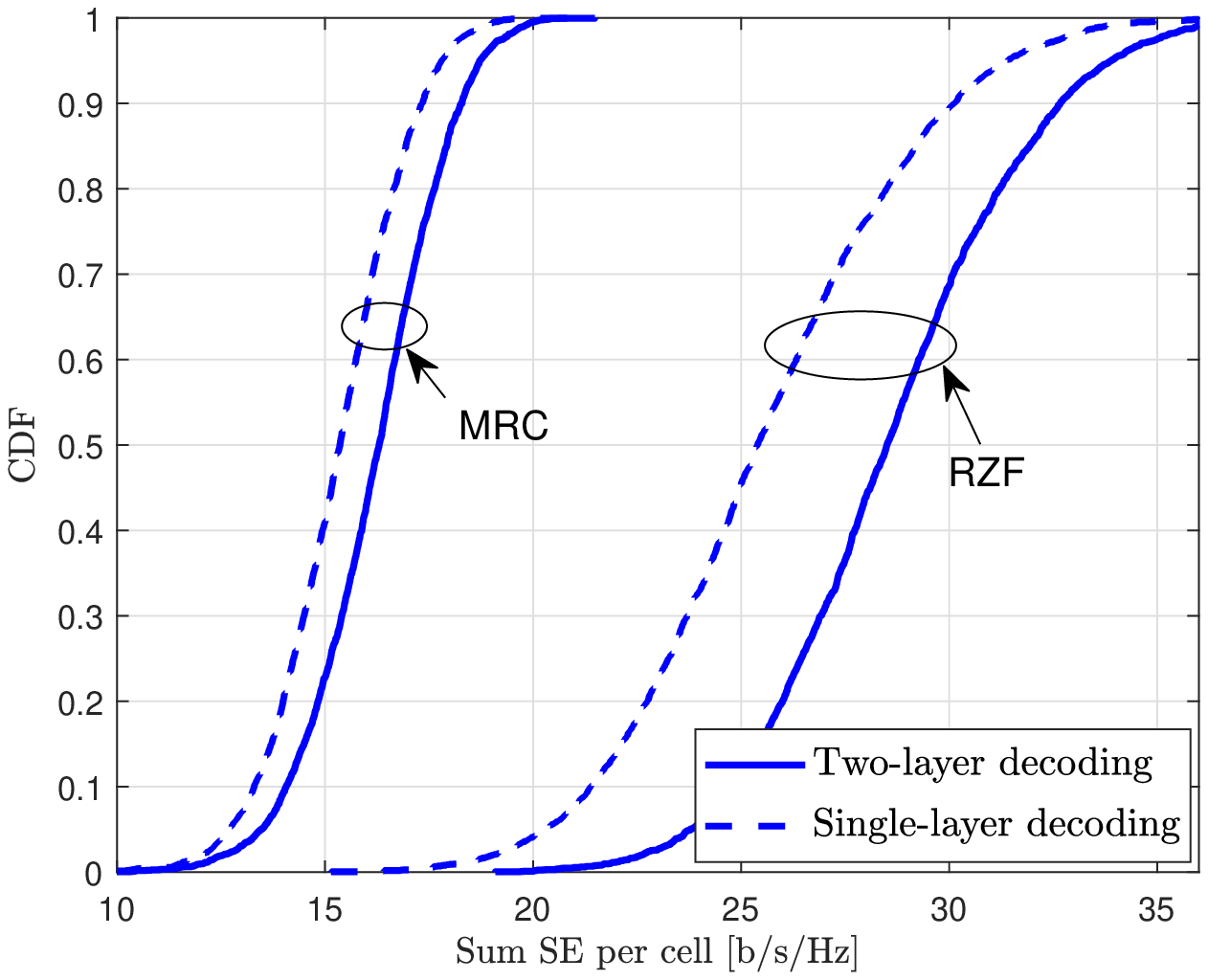} \vspace*{-0.55cm}
			\caption{\abb{CDF} of sum \abb{SE} per cell for \abb{MRC} and \abb{RZF} with $M=200, K =5,\varsigma = 0.5$.}
			\label{Fig-RZF}
			\vspace*{-0.3cm}
		\end{minipage}
		\vspace*{-0.05cm}
		\end{figure*}
		
		 We consider a wrapped-around cellular network with four cells.  The distance between user~$k$ in cell~$l'$ and BS~$l$ is denoted by $d_{l',k}^l$.  
		The users in each cell are uniformly distributed over the cell area that is at least 35\,m away from the BS, i.e., $d_{l',k}^l \geq \text{35\,m}$.
		Monte-Carlo simulations are carried out over $300$ random sets of user locations. We model the system parameters and large-scale fading similar to the 3GPP LTE specifications \cite{LTE2010b}. The system uses $20$\,MHz of bandwidth, the noise variance is $-96$\,dBm, and the noise figure is $5$\,dB.  The large-scale fading coefficient $\beta_{l,k}^{l'}$ is computed in decibel scale as $\beta_{l,k}^{l'}= -148.1 - 37.6 \log_{10}( d_{l,k}^{l'} / 1\,\text{km} )+ z_{l,k}^{l'},$
		where the decibel value of the shadow fading, $z_{l,k}$, has a Gaussian distribution with zero mean and standard derivation $7$. The spatial correlation matrix of the channel from user~$k$ in cell~$l$ to BS~$l'$ is described by the exponential correlation model, that models a uniform linear array with the correlation magnitude $\varsigma \in [0, \, 1]$ \cite{Loyka2001a}. {The correlation magnitude is multiplied with a unique phase-shift in every correlation matrix, selected as the user's incidence angle to the array}. We assume that the power is fixed to $200$\,mW for each pilot symbol and it is also the maximum power that each user can allocate to a data symbol, i.e., $P_{\max,l,k} = 200$\,mW. The following methods are compared in the simulation:
		\begin{itemize}
			\item[\textit{(i)}] \textit{Single-layer decoding system with fixed data power}: Each \abb{BS} uses \abb{MRC} for data decoding for the users in the own cell, and all users transmit data symbols with the same power $200$\,mW.  
			
			\item[\textit{(ii)}] \textit{Single-layer decoding system with data power control}: This benchmark is similar to \textit{(i)}, but the data powers are optimized using a modified version of Theorem~\ref{Theorem: IterativeSol}.
			
			\item[\textit{(iii)}] \textit{Two-layer decoding system with fixed data power and \abb{LSFD} vectors}: The network deploys the two-layer decoding as shown in Fig.~\ref{fig:decoder}, using \abb{MRC} and \abb{LSFD}. The data symbols are transmitted at the maximum power $200$\,mW and the \abb{LSFD} vectors are computed using Corollary~\ref{corollary:Opt_LSFD}. 
			
			\item[\textit{(iv)}] \textit{Two-layer decoding system with optimized data power and  \abb{LSFD} vectors}: This is the proposed method, where the data powers and \abb{LSFD} vectors are computed using the weighted \abb{MMSE} algorithm as in Theorem~\ref{Theorem: IterativeSol}.
		\end{itemize}
		
		Fig.~\ref{Fig-Convergence} shows the convergence of the proposed method for sum \abb{SE} optimization in Theorem~\ref{Theorem: IterativeSol}. 
		{From the initial random data powers, uniformly distributed in the feasible set}, updating the optimization variables gives improved sum \abb{SE} in every iteration. 
		For the two layer case \textit{(iv)}, the sum \abb{SE} per cell is about $22.2\%$ better at the stationary point than at the initial point. At convergence, \textit{(iv)} gives $2.4\%$ better sum \abb{SE} than \textit{(ii)}. The proposed optimization methods need around $100$ iterations to converge, but the complexity is low since every iteration in the algorithm consists of evaluating a closed-form expression.

		Fig.~\ref{Fig-MMSE-DiffCofact} shows the sum \abb{SE} per cell as a function of the channel correlation magnitude $\varsigma$ for a multi-cell Massive \abb{MIMO} system. 
		First, we observe the substantial gains in sum \abb{SE} attained by using \abb{LSFD}. The sum \abb{SE} increases with up to $7.5\%$ in the case of equally fixed data powers, while that gain is about $7.9\%$ for jointly optimized data powers and \abb{LSFD} vectors. Moreover, this figure shows that the performance is greatly improved when the data powers are optimized. The gain varies from $17.9\%$ to $20.3\%$. The gap becomes larger as the channel correlation magnitude increases. This shows the importance of doing joint data power control and \abb{LSFD} optimization in Massive \abb{MIMO} systems with spatially correlated channels.
		 
Fig.~\ref{Fig-RZF} compares the cumulative distribution function (\abb{CDF}) of the sum \abb{SE} per cell with either \abb{MRC} or \abb{RZF} in the first layer, where the latter requires the use of the new general SE expression in Theorem~\ref{Theorem1v1}. An equal power $200$ mW is allocated to each transmitted symbol. Because \abb{RZF} mitigates non-coherent interference effectively in the first layer, the second layer can increase the average SE by $11.80\%$. If \abb{MRC} is used in the first layer, the SE gain from using \abb{LSFD} using is only $5.84\%$. At the $95\%$-likely point, the two layer decoding system outperforms the single layer counterpart by $4.71\%$ and $17.35\%$ when using \abb{MRC} or \abb{RZF}, respectively.		
		
		\section{Conclusion} \label{Section:Conclustion}
		This paper has investigated the ability of \abb{LSFD} to mitigate inter-cell interference in multi-cell Massive \abb{MIMO} systems with  spatially correlated Rayleigh fading. \abb{LSFD} is a two-layer decoding method, where a second decoding layer to mitigate inter-cell interference is applied after the classical decoding. We derived new SE expressions support arbitrary spatial correlation and first-layer decoding. We used these expressions to optimize the data powers and \abb{LSFD} vectors, to maximize the sum \abb{SE} of the network. Even though the sum \abb{SE} optimization is a non-convex and NP-hard problem, we proposed an iterative approach to obtain a stationary point with low computational complexity. 
		Numerical results demonstrate the effectiveness of \abb{LSFD} in reducing pilot contamination with the improvement of sum \abb{SE} for each cell about $17\%$ in the tested scenarios, while the optimized data power control and \abb{LFSD} design can improve the sum \abb{SE} with more than $20\%$. The gains are larger when using \abb{RZF}  in the first layer than when using \abb{MRC}.
		
		\bibliographystyle{IEEEtran}
		\bibliography{IEEEabrv,refs}
	\end{document}